\documentclass[11pt]{article}
\usepackage{bbm}
\usepackage{hyperref}
\usepackage{amsmath}
\RequirePackage{amssymb}
\usepackage{colortbl,xcolor}
\usepackage{amsthm, verbatim}
\usepackage{array}

\colorlet{optimal}{black!30}
\colorlet{best}{black!20}
\colorlet{better}{black!10}

\newcommand{\EE}{\mathcal{E}}
\newcommand{\G}{\mathcal{G}}

\newcommand{\PP}{\mathbb{P}}

\newcommand{\st}{\text{ s.t. }}
\newcommand{\defeq}{\,{\buildrel \triangle \over =}\,}
\newcommand{\Z}{\mathbb{Z}}
\newcommand{\1}{\mathbbm{1}}
\newcommand{\bo}[1]{O\left(#1\right)}
\newtheorem{prop}{Proposition}
\newtheorem{theorem}{Theorem}
\newtheorem{lemma}{Lemma}
\newtheorem{definition}{Definition}
\theoremstyle{remark}
\newtheorem*{remark}{Remark}

\title{A Message Passing Algorithm for the Problem of Path Packing in Graphs}
\author{Patrick Eschenfeldt and David Gamarnik}

\begin{document}

\maketitle
\abstract{We consider the problem of packing node-disjoint directed paths
in a directed graph. We consider a variant of this problem where each path 
starts within a fixed subset of root nodes, subject to a given bound on the
length of paths. This problem is motivated by the so-called kidney exchange
problem, but has potential other applications and is interesting in its own
right. 

We propose a new algorithm for this problem based on the message passing/belief
propagation technique. A priori this problem does not have an associated
graphical model, so in order
 to apply a belief propagation algorithm we provide a novel 
representation of the problem as a graphical model.
Standard belief propagation on this model has poor scaling behavior, so we
provide an efficient implementation that significantly decreases the
complexity. We provide numerical results comparing the performance of our
algorithm on both artificially created graphs and real world networks to
several alternative algorithms, including algorithms based on integer
programming (IP) techniques. These comparisons show that our algorithm scales
better to large instances than IP-based algorithms and often finds better
solutions
than a simple algorithm that greedily selects the longest path from each root
node. In some cases it also finds better solutions than the ones found by
IP-based algorithms even when the latter are allowed
to run significantly longer than our algorithm.}

\section{Introduction.}

In this paper we consider the problem of packing node-disjoint directed
paths into a directed graph, with each path starting within a designated subset
of ``root'' nodes. For a given maximum path length, our objective is to
include as many nodes as possible in paths.
This problem is motivated by the kidney exchange problem
(KEP) studied in \cite{Anderson20012015}, in which donors and
recipients must be matched together to maximize the number of donations.
In some situations these donations are performed in chains starting from a lone
donor, and logistical concerns encourage limiting the length of chains. 
This problem also has other potential
applications such as coordinating chained transactions,
and is an interesting and hard combinatorial optimization problem in general. 

We present a new algorithm for this
problem based on the message passing/belief propagation technique (BP for
short). This algorithm depends on a new representation of the problem
that allows us to create a graphical model which embeds the optimization
problem. After describing this representation we will show that a direct
application of standard BP results in poor scaling, so we provide an efficient
implementation that provides significant scaling advantages. Letting $n$ be the
number of nodes in the graph, $K$ be the
maximum path length, and $\Delta$ be the maximum degree, the efficient
implementation can perform a full iteration of the algorithm in
$\bo{n\Delta K}$ computations, while a direct implementation of belief
propagation requires $\bo{n^2\Delta^6K^2}$. 

We compare our algorithm to several alternatives and show that it
consistently scales well
 and often finds better solutions than other fast algorithms. 
In particular, we compare our efficiently implemented BP to a simple
greedy algorithm, called Greedy from now on for simplicity, and an existing 
Integer Programming (IP) algorithm used in \cite{Anderson20012015} which we
call KEP. The Greedy algorithm which we consider
makes greedy selections only between constructed
paths and not during the construction of individual paths, as at every root 
node it searches exhaustively for the longest path which can be added to the
solution starting from that node. We also use our new representation of the
problem to derive a secondary novel algorithm based on IP, which we call
Parent-Child-Depth (PCD).
Comparisons are made on a variety of randomly created graphs as well as on
several real world networks. Our random graphs are generated with each edge
existing
independently with probability $c / n$ for various values of $c$. We find that
for $n = 1000$ and $n = 10000$ with a path length bound of $5$ nodes, with 
10\% to 25\% of nodes designated as root nodes, and $c$ between 2 and 5, our
message passing algorithm consistently finds more nodes than Greedy while
running in comparable time. In many cases BP and Greedy run significantly
faster than the IP algorithms, which are cut off after a time limit if they
have not yet found the optimal solution, and in several of these cases BP also
finds solutions with more nodes than any other algorithm. We find that BP also
performs well relative to alternative algorithms when the path length bound is
increased to 10 or 15 nodes. In the regime where paths as long as 15 nodes are
allowed, PCD often finds the optimal solution in reasonable time while KEP
almost never does. We also provide numerical results for real world networks,
with sizes ranging from 5518 nodes to 260982 nodes. In all but the smallest
network the IP algorithms produce no useful results, and in most graphs BP
finds better solutions on average with respect to a random choice of root nodes
than Greedy, while running in comparable time.

Our approach to this problem is motivated by previous work applying message
passing techniques to the problem of prize-collecting Steiner trees (PCST)
done by Bayati, et.~al.~\cite{Bayati-algorithm} (also used by Bailly,
et.~al.~\cite{Bailly}). In particular, the PCST problem does not have a natural
representation as a graphical model so the authors in \cite{Bayati-algorithm}
design a new representation of the problem that leads to a graphical model and
then show how to implement BP on that model. This paper uses a similar approach
for the path packing problem.

In another relevant paper, Altarelli,
et.~al.~\cite{altarelli} apply message passing techniques to edge-disjoint
paths. They directly derive a message passing algorithm for general graphs from
an iterative cost calculation scheme valid for trees and apply it to locally
tree-like graphs and find that their algorithm performs better in terms of
paths found than various alternatives. Unlike this work, in our model paths may
end at any non-root node (as we do not consider a lone root node to constitute
a path in itself, every path must include at least one edge and thus one
non-root node).

Another example of applying a message passing approach while making major
modifications to standard methods is \cite{eschenfeldt}. The
authors introduce a new type of graphical model, called a memory factor
network, along with a message passing style algorithm to perform inference on
such a model. Designing the message passing algorithm to match the structure of
the problem provides efficiency advantages over variants of belief propagation
and also provides convergence guarantees.

The remainder of this paper is organized as follows: In section \ref{sec:setup}
we describe the problem formulation, in section \ref{sec:algorithm} we detail
the core of our message passing algorithm, in section \ref{sec:efficient} we
detail our method for efficiently implementing this algorithm, in section
\ref{sec:reconstruction} we discuss how to translate the result of the message
passing into a solution to the original problem, in section
\ref{sec:algorithms} we discuss alternative algorithms, and in section
\ref{sec:performance} we provide numerical results comparing our algorithm to
these alternatives.

\section{Problem setup.}\label{sec:setup}
The problem description consists of a directed graph $\G$ with node set $[n]
= \{1,2,\ldots,n\}$, an edge set $\EE$, and a positive integer $K$.  Let
$\EE'$ be the undirected version of the edge set of $\G$. We assume that $\G$
contains no isolated nodes, so for every $i \in [n]$ there exists some $j \in
[n]$ such that $(i,j) \in \EE'$. The nodes $[n]$ are partitioned into subsets
$U$ and $V$, and the edges are elements of $U \times V$ and $V \times V$ only.
In particular, no edges start in $V$ and end in $U$. We define a \emph{directed
simple path} in a directed graph as a sequence of distinct nodes
$\{i_1,i_2,\ldots,i_k\}$ such that $(i_l, i_{l+1}) \in \EE$ for all $l =
1,\ldots,k-1$. All paths considered in this paper will be directed simple
paths. For a path $P = \{i_1,\ldots,i_k\}$ we define the \emph{length} of $P$, 
denoted $|P|$, to be the number of nodes in the path, namely $k$. Our goal is
to find a collection $\mathcal{P} = \{P_1,\ldots,P_k\}$ of node-disjoint paths
with each path $P_j$ starting with a node in $U$ and each having length at most
$K$. Specifically, we define
\[v(\mathcal{P}) = \sum_{P \in \mathcal{P}} |P|\]
to be the count of the nodes in a collection of paths. Let $\Pi$ be the set
of feasible collections $\mathcal{P}$. Then our problem is
\begin{equation}\label{eq:opt-prob-orig}
	\max_{\mathcal{P} \in \Pi} v(\mathcal{P}).
\end{equation}
We now introduce a formalization for the problem \eqref{eq:opt-prob-orig} that
is conducive to message passing.

We will use special symbols $*$ and
$\bullet$ to represent certain relationships to be described below. We will
also introduce the notation $\partial i$ to represent the set of nodes
adjacent to the node $i$ in the undirected version of the graph $\G$.
That is, $\partial i$ is the set of nodes $j \in [n]$ such that $(i,j)
\in \EE$ or $(j,i) \in \EE$. We also let $\partial^X i =
\partial^i \cap X$ and let $\partial i \setminus j$ stand for the more formally
correct $\partial i \setminus \{j\}$.

For each node $i \in [n]$, we introduce variables $(d_i,p_i,c_i)$ where $d_i
\in [K] \cup \{*\}$, $p_i \in \partial i \cup \{*,\bullet\}$, $c_i \in
\partial^V i \cup \{*,\bullet\}$. These variables intend to represent the
depth of the node ($d_i$), its parent ($p_i$), and its child ($c_i$) in a
given solution. The special value $*$ indicates that the node does not
participate in a path, in which case we will have $d_i = p_i = c_i = *$. The
special value $\bullet$ indicates that the node is at one end of a path. Thus
if $p_i = \bullet$ this means node $i$ does not have a parent and thus starts a
path, whereas if $c_i = \bullet$ node $i$ does not have a child and thus ends a
path. Note that if $c_i$ is a node it must be a node in $V$ because nodes in
$U$ must either act as the root of a path or not participate in a path.

To formulate the optimization problem, we need to define a feasible set and an
objective function. We begin with the feasible set. We first define a function
that acts as an indicator for the consistency of the variables $(d_i,p_i,c_i)
$ at a single node $i$. For $i \in [n]$, $d_i \in [K] \cup \{*\}$,
$p_i \in \partial i \cup \{*,\bullet\}$, $c_i \in \partial^V i \cup
\{*,\bullet\}$, let
\begin{align}
	f_i(d_i,p_i,c_i) &\defeq \1\{d_i = c_i = p_i = *\} 
\label{eq:node-valid-def-np}\\
	&\quad + \1\{p_i = \bullet, c_i \not\in \{*,\bullet\}, d_i = 1, i \in U\}
\label{eq:node-valid-def-root}\\
	&\quad + \1\{p_i \not\in \{*,\bullet\}, c_i \not\in \{*,p_i\}, d_i \not\in
\{*,1\}, i \in V\}.\label{eq:node-valid-def-non-root}
\end{align}
Term \eqref{eq:node-valid-def-np} is the case in which node $i$ does not
participate, term \eqref{eq:node-valid-def-root} is the case in which node $i$
is the root of a path, and term \eqref{eq:node-valid-def-non-root} is all cases
where node $i$ participates in a path but is not the root. Note that these
three cases are mutually exclusive, so $f_i$ is an indicator function. It is
equal to zero if the variables at node $i$ are not consistent with any
global configuration $(d,p,c) = (d_j,p_j,c_j)_{j \in [n]}$ representing a set
of valid paths for the graph $\G$, and equal to one otherwise.

For each pair of nodes $i$ and $j$ that are adjacent in the undirected version
of $\G$, we define a function $g_{ij}(d_i,p_i,c_i,d_j,p_j,c_j)$ of the
variables $(d_i,p_i,c_i,d_j,p_j,c_j)$ for the nodes $i$ and $j$ which is an
indicator that these variables are consistent with each other and the existence
or nonexistence of the edges $(i,j)$ and $(j,i)$. Namely, they correspond to a
locally valid configuration, so if, e.g., node $i$ reports that its parent is
node $j$, then the edge $(j,i)$ exists, node $j$ reports that its child is node
$i$, and we have $d_i = d_j + 1$. Formally, for $i \in [n]$, $j
\in \partial i$, $d_i \in [K] \cup \{*\}$, $p_i \in \partial i \cup
\{*,\bullet\}$, $c_i \in \partial^V i \cup \{*,\bullet\}$, $d_j
\in [K] \cup \{*\}$, $p_j \in \partial j \cup \{*,\bullet\}$, $c_j \in
\partial^V j \cup \{*,\bullet\}$, let 
\begin{align}
	g_{ij}(d_i,p_i,c_i,d_j,p_j,c_j) &\defeq \notag\\
	&\hspace{-0.21\textwidth}\1\{p_i = j, c_j = i, p_j \ne i, c_i \ne j, (j,i) \in
\EE, d_i = d_j + 1\} \label{eq:edge-valid-def-j-parent}\\
	&\hspace{-0.25\textwidth}+ \1\{p_j = i, c_i = j, p_i \ne j, c_j \ne i, (i,j)
\in \EE, d_j = d_i + 1\} \label{eq:edge-valid-def-i-parent} \\
	&\hspace{-0.25\textwidth}+ \1\{p_i \ne j, c_j \ne i, p_j \ne i, c_i \ne j\}.
\label{eq:edge-valid-def-neither}
\end{align}
Lines \eqref{eq:edge-valid-def-j-parent} and \eqref{eq:edge-valid-def-i-parent}
capture the cases where $j$ is the parent of $i$ and $i$ is the parent of $j$,
respectively, and enforce the relationships that must occur in those cases.
Line \eqref{eq:edge-valid-def-neither} captures all cases where there is no
direct connection between nodes $i$ and $j$. All three cases are mutually
disjoint, so $g_{ij}$ is an indicator function. In summary, $g_{ij}$ is equal
to one if parent/child relationships agree, the edges $(i,j)$ and $(j,i)$ are
used only in the appropriate directions, and depth relationships are
consistent.

Finally we define a indicator function for variable consistency at both the
node and edge level, which we define for each $i \in [n]$, $j \in \partial i$
as
\begin{equation}\label{eq:node-edge-valid}
h_{ij}(d_i,p_i,c_i,d_j,p_j,c_j) \defeq
g_{ij}(d_i,p_i,c_i,d_j,p_j,c_j)f_i(d_i,p_i,c_i)f_j(d_j,p_j,c_j).
\end{equation} 
With this definition of $h_{ij}$, we can define the set
\begin{equation}\label{eq:feasible-set-def}
M \defeq \big\{(d,p,c): h_{ij}(d_i,p_i,c_i,d_j,p_j,c_j) = 1\ \forall\ i\in
[n], j \in \partial i\big\}
\end{equation}
which will serve as the feasible set for our optimization. 

Note that each collection $\mathcal{P} \in \Pi$ of node disjoint paths is
associated with a unique assignment of variables $\{(d_i,p_i,c_i)\}_{i
\in [n]}$, which we will denote by $(d(\mathcal{P}), p(\mathcal{P}),
c(\mathcal{P}))$. In fact the converse is also true:

\begin{prop}
For every $(d,p,c) \in M$ there exists a unique $\mathcal{P} \in \Pi$ such that
$(d,p,c) = (d(\mathcal{P}), p(\mathcal{P}), c(\mathcal{P}))$.
\end{prop}

\begin{proof}
Given a set of variables $(d,p,c) \in M$, recall that for every $i \in [n]$ we
have $f_i(d_i,p_i,c_i) = 1$ and for every $j \in \partial i$ we have
$g_{ij}(d_i,p_i,c_i,d_j,p_j,c_j) = 1$.

We will construct a feasible set of node-disjoint paths by the following
procedure. Find the set of nodes $i$ such that $d_i = 1$, which we denote
$\mathcal{R}$.

If $\mathcal{R}$ is empty then for every $i \in U$ we must be in case
\eqref{eq:node-valid-def-np} because case \eqref{eq:node-valid-def-root} is
excluded. Thus $d_i = p_i = c_i = *$. For $i \in V$, by
\eqref{eq:node-valid-def-np} and \eqref{eq:node-valid-def-non-root} we either
have $d_i = *$ or $d_i > 1$. But by \eqref{eq:edge-valid-def-j-parent} $d_i =
2$ only if $d_j = 1$ for some $j \in \partial i$. This implies that $d_i \ne 2$
for any $i \in V$. Repeating this argument implies $d_i \ne l$ for any $l \in
[K]$, so we conclude $d_i = *$. Then \eqref{eq:node-valid-def-np} implies
$p_i = c_i = *$. Thus for every node $i \in [n]$ we have  $d_i = p_i = c_i =
*$, so we let $\mathcal{P} = \emptyset$ because $d_i(\emptyset) =
p_i(\emptyset) = c_i(\emptyset) = *$ for all $i \in [n]$.
 
Otherwise, if $\mathcal{R}$ is nonempty, we will construct one path for each $i
\in \mathcal{R}$. Consider some $i_1 \in [n]$ such that $d_{i_1} = 1$.
By \eqref{eq:node-valid-def-root} we must have $i_1 \in U$, so it will be valid
to start a path at $i_1$. Again by
\eqref{eq:node-valid-def-root} we know $c_{i_1} \ne *,\bullet$, so $c_{i_1} =
i_2$ for some $i_2 \in \partial^V i$. By \eqref{eq:edge-valid-def-i-parent}
with $i = i_i$ and $j = i_2$ we have $d_{i_2} = d_{i_1} + 1 = 2$, $(i_1,i_2)
\in \EE$, $p_{i_2} = i_1$ and $c_{i_2} \ne i_1$. Thus it is valid to
start our path with the node sequence $\{i_1,i_2\}$. If $c_{i_2} = \bullet$
then we end this path and add $P = \{i_1,i_2\}$  to $\mathcal{P}$. If $c_{i_2}
\ne \bullet$ then by $d_{i_2} = 2$ and \eqref{eq:node-valid-def-non-root} we
have $c_{i_2} = i_3$ for some $i_3 \in \partial^V i_2$. By
\eqref{eq:edge-valid-def-i-parent} with $i = i_2$ and $j = i_3$ we have
$d_{i_3} = 3$, $(i_2,i_3) \in \EE$, $p_{i_3} = i_2$, and $c_{i_3} \ne i_2$.
Thus it is valid to start our path with the node sequence $\{i_1,i_2,i_3\}$. If
$c_{i_3} = \bullet$ we end this path and add it to $\mathcal{P}$. Otherwise we
iterate, because \eqref{eq:node-valid-def-non-root} implies $c_{i_3} = i_4$ for
some $i_4 \in \partial^V i_3$. In the generic step, we are considering
$c_{i_{l-1}} = i_l$ for some $i_l \in \partial^V i_{l-1}$ and we have
$d_{i_{l-1}} = l-1$. Then \eqref{eq:edge-valid-def-i-parent} for $i = i_{l-1}$
and $j = i_l$ implies $d_{i_l} = l$, $(i_{l-1},i_l) \in \EE$, $p_{i_l} =
i_{l-1}$ and $c_{i_l} \ne i_{l-1}$. If $c_{i_l} = \bullet$ we terminate the
path and add $P = \{i_1,i_2,\ldots,i_l\}$ to $\mathcal{P}$. Otherwise there
exists some $i_{l+1} \in \partial^V i_l$ such that $c_{i_l} = i_{l+1}$ so we
iterate. This process is guaranteed to terminate because $d_{i_l} = l$ and
$d_{i_l} \le K$. Once we have added the path starting at node $i_1$ we repeat
the process for each node $i$ with $d_i = 1$.

To see that $(d(\mathcal{P}), p(\mathcal{P}), c(\mathcal{P})) = (d,p,c)$ after
this process, first consider a node $j$ which does not participate in any path
in $\mathcal{P}$. We have $d_j(\mathcal{P}) = p_j(\mathcal{P}) = 
c_j(\mathcal{P}) = *$ so we want to show $d_j = p_j = c_j = *$. Because every
node $i$ with $d_i = 1$ starts a path in the above procedure, we must have $d_j
\ne 1$. Suppose $d_j = l$ for some $l \in \{2,\ldots,K\}$. Then
\eqref{eq:node-valid-def-non-root} implies $p_j \ne \{*,\bullet\}$ so there
exists some $i \in \partial j$ such that $p_j = i$ and
\eqref{eq:edge-valid-def-i-parent} implies $d_i = l - 1$. If $i$ participates
in a path in $\mathcal{P}$ then our above procedure would also include $j$. If
$i$ does not participate then we have found a non-participating node with depth
$l - 1$. Thus if there are were any non-participating $j$ with $d_j \in [K]$ 
there would be a non-participating $j$ with $d_j = 1$, but we have already
shown there is no such $j$. We conclude that $d_j = *$ for all
non-participating $d_j$. Then \eqref{eq:node-valid-def-np} implies $p_j = c_j =
*$, as desired.

Next consider a node $i_l$ which participates in a path in $\mathcal{P}$ at
depth $l$. 
If $l = 1$ then $i_l$ starts a path and has no parent and (by the
above procedure) has child $c_{i_l}$. Thus we have $d_{i_l}(\mathcal{P}) = 1$,
$p_{i_l}(\mathcal{P}) = \bullet$, and $c_{i_l}(\mathcal{P}) = c_{i_l}$. The
only nodes $i_l$ that start paths are those with $d_{i_l} = 1$ and
\eqref{eq:node-valid-def-root} implies $p_{i_l} = \bullet$, so we have
$(d_{i_l}(\mathcal{P}),p_{i_l}(\mathcal{P}), c_{i_l}(\mathcal{P})) =
(d_{i_l},p_{i_l},c_{i_l})$, as desired.
If $l > 1$ then $i_l$ has $d_l = l$,
parent $i_{l-1} = p_{i_l}$ (as noted
during the procedure above) and either ends the path or has a child $i_{l+1} =
c_{i_l}$. By the procedure the path ends only if $c_{i_l} = \bullet$, so we
have $c_{i_l}(\mathcal{P}) = \bullet = c_{i_l}$. Thus we have
$(d_{i_l}(\mathcal{P}),p_{i_l}(\mathcal{P}), c_{i_l}(\mathcal{P})) =
(d_{i_l},p_{i_l},c_{i_l})$, as desired. If the path continues it continues to
$i_{l+1} = c_{i_l}$ so $c_{i_l}(\mathcal{P}) = c_{i_l}$, and again we have 
$(d_{i_l}(\mathcal{P}),p_{i_l}(\mathcal{P}), c_{i_l}(\mathcal{P})) =
(d_{i_l},p_{i_l},c_{i_l})$, as desired.
\end{proof}

Because we have a bijection between feasible collections of node-disjoint paths
in $\G$ and the set $M$, we can treat $M$ as the feasible region for our
optimization problem.

To define our objective, we define the function
\begin{equation}\label{eq:node-cost-def}
\eta(d_i,p_i,c_i) \defeq \begin{cases} 0 & d_i = p_i = c_i = * \\ 1 &
\text{otw} \end{cases}
\end{equation}
which assigns value 0 to the case where the node $i$ does not participate in a
path and value 1 to all other cases. We further define the function
\begin{equation}\label{eq:graph-cost-def}
H(d,p,c) \defeq \sum_{i \in [n]}\eta(d_i,p_i,c_i)	
\end{equation}
that sums the value of each node across the entire graph. Note that when
applied to a valid $(d,p,c)$, i.e.~an element of the feasible set $M$, $H$
simply counts the number of nodes of the graph which participate in some
path.

Thus we can reformulate the optimization problem \eqref{eq:opt-prob-orig} as
\begin{equation}\label{eq:opt-prob}
	\max_{(d,p,c) \in M} H(d,p,c).
\end{equation}

\section{Algorithm.}\label{sec:algorithm}

Belief propagation algorithms are used to estimate either the marginal
distributions of variables or the solution to the \emph{maximum a
posteriori} (MAP) problem in graphical models. We will use the framework
corresponding to the latter option. The MAP solution for a probability
distribution is the most likely joint assignment of all the variables in the
distribution. To utilize this approach, we transform the optimization problem
\eqref{eq:opt-prob} into a problem of finding a MAP assignment for a
probability distribution. To achieve this transformation we introduce a
corresponding undirected graphical model. A maximum likelihood assignment of
all variables in this distribution will be an optimal solution to
\eqref{eq:opt-prob}.

To create our undirected graphical model, we introduce node potentials
\begin{equation}
\phi_i(d_i,p_i,c_i) = \begin{cases} e^{-\beta} & d_i = p_i = c_i = * \\ 1 &
\text{otw} \end{cases}\label{eq:node-potential}
\end{equation}
where $\beta > 0$ is a system parameter to be chosen. We use $h_{ij}$
as defined in \eqref{eq:node-edge-valid} as our edge potentials. Together these
define the probability distribution $\PP$
\[\PP(d,p,c) \propto \prod_{i \in [n]}\phi_i(d_i,p_i,c_i)\prod_{(i,j) \in
\EE'}h_{ij}(d_i,p_i,c_i,d_j,p_j,c_j).\]
This probability distribution assigns zero probability to any global
configuration $(d,p,c)$ which does not correspond to an element of $M$, and
among positive probability configurations it assigns higher probability to
those which include more nodes.

We can now solve our original optimization problem by solving the MAP problem
on this undirected graphical model. Note that we may have multiple solutions.

The core component of the algorithm is the set of messages, which we define for
each ordered pair $(j,i)$ with $j \in [n]$ and $i \in \partial j$, representing
the message from node $j$ to node $i$. Thus each undirected edge $(i,j)
\in \EE'$ is associated with two sets of messages, representing the information
sent from $i$ to $j$ and from $j$ to $i$, respectively. The message from $j$ to
$i$ is denoted $b_{j \to i}(d_i,p_i,c_i)$ for all values of $d_i$, $p_i$,
and $c_i$ in the previously specified ranges  $d_i \in [K] \cup \{*\}$, $p_i
\in \partial i \cup \{*,\bullet\}$, $c_i \in \partial^V i \cup \{*,\bullet\}$.
Our approach to this problem will be to implement parallel belief propagation
in the standard Min-Sum form (which can be derived from the Max-Product form by
taking the negative log of all messages;  see, e.g., \cite{Koller-Friedman} or
\cite{Yedidia:2003:UBP:779343.779352}), with messages sent according to the
update rule
\begin{align}\label{eq:min-sum-update}
b_{j \to i}(d_i,p_i,c_i) &= \min_{(d_j,p_j,c_j)}\bigg(-\log\phi_j(d_j,p_j,c_j)
\notag \\
&\hspace{0.15\textwidth} - \log h_{ij}(d_i,p_i,c_i,d_j,p_j,c_j) \\
&\hspace{0.15\textwidth} + \sum_{k \in \partial j \setminus i} b_{k \to
j}(d_j,p_j,c_j)\bigg). \notag
\end{align}
We use the convention $-\log(0) = \infty$. In other words, if $h_{ij} = 0$ for
some configuration then that configuration will not be included in the
minimization. Otherwise the $\log(h_{ij})$ term contributes $\log(1) = 0$.

Because our graphical model potentially has cycles, we must compute these
messages iteratively from some starting values, which we set to be $b_{j \to
i}(d_i,p_i,c_i) = 1$ for all $j \in [n]$, $i \in \partial j$ and all values of 
$d_i,p_i,c_i$, with the goal of finding a fixed point solution. From the
messages we can compute the max-marginals at each node as
\begin{equation}\label{eq:max-marginals}	
\bar{p}_i(d_i,p_i,c_i) = \exp\left(\log \phi_i(d_i,p_i,c_i) - \sum_{k \in
\partial i} b_{k \to i}(d_i,p_i,c_i).\right)
\end{equation}
These max-marginals are, according to the Belief Propagation approach, intended
to approximate the probability of a maximum-probability configuration of the
whole system given the arguments $d_i,p_i,c_i$ are fixed. As such, they can be
used in reconstructing the maximum global configuration. The non-uniqueness of
solutions to \eqref{eq:opt-prob} complicates this, however, which we discuss
below in section \ref{sec:reconstruction}.

The update rule \eqref{eq:min-sum-update} and max-marginal equation
\eqref{eq:max-marginals} characterize the entirety of our algorithm. For some a
priori maximum number of iterations $T$ and initial messages $b_{j \to
i}^0(d_i,p_i,c_i) = 1$ for all $j \in [n]$, $i \in \partial j$ and all values
of $d_i,p_i,c_i$ we compute $b^t$ from $b^{t-1}$ according to
\eqref{eq:min-sum-update} for $t \in [T]$. If the messages reach a fixed point,
i.e.~$b^t = b^{t+1}$ for some $t < T$, we halt iterations at that point.
Furthermore, for any $0 < t \le T$ we can use $b^t$ and
\eqref{eq:max-marginals} to compute estimates of the max-marginals for any node
in any configuration at that step of the algorithm.

In this form, the number of messages is $\bo{K \sum_{i \in
[n]}\Delta_i^3}$ where $\Delta_i = |\partial i|$ is the degree of node $i$ in
the undirected graph. Indeed, for each node $i \in [n]$ there are $\Delta_i$
incoming messages, and for each message there are $K$ choices of $d_i$ and
$\bo{\Delta_i^2}$ choices of $(p_i,c_i)$ that define messages.  For $\Delta =
\max_{i \in [n]} \Delta_i$ the number of messages is
\begin{equation}\label{eq:orig-messages}
	\bo{n\Delta^3K}.
\end{equation}
At each iteration of the algorithm a message from $j$ to $i$ must be updated by
computing a minimization over all $\bo{K\Delta_j^2}$ choices of the variables
$(d_j,p_j,c_j)$. The argument of the minimization includes a sum over the
neighbors of $j$, which can be performed once for each $j \in [n]$ for a total
computational cost $\bo{n}$ per iteration. Thus a full iteration of all
messages requires $\bo{n K^2 \sum_{i \in [n]}\left( \Delta_i^3\sum_{j \in
\partial i}\Delta_j^2\right)}$ computations. For $\Delta = \max_{i
\in [n]}\Delta_i$, the number of computations is
\begin{equation}\label{eq:orig-computations}
	\bo{n^2\Delta^6K^2}.
\end{equation}
Rather than implementing this scheme directly we will instead introduce a more
efficient representation of the messages, which is described in the next
section.

\section{Efficient representation.}
\label{sec:efficient}

In this section we introduce a more efficient representation of the message
passing algorithm described above. This representation will require only
$\bo{n K \Delta}$ messages compared to $\bo{n K \Delta^3}$ for the original
description \eqref{eq:min-sum-update}, and will require $\bo{n K \Delta}$
computations to be performed per iteration compared to $\bo{n^2K^2\Delta^6}$
computations for the original representation.

Before defining the messages for the new representation, we now introduce an
intermediate set of messages and rewrite  \eqref{eq:min-sum-update} as two
steps, with
\begin{align}
b_{j \to i}(d_i,p_i,c_i) &= \min_{\substack{d_j,p_j,c_j \st \\
h_{ij}(d_i,p_i,c_i,d_j,p_j,c_j) = 1}}\psi_{j \to i}(d_j,p_j,c_j),
\label{eq:update-b} \\
\psi_{j \to i}(d_j,p_j,c_j) &= -\log\phi_j(d_j,p_j,c_j) + \sum_{k \in \partial
j \setminus i} b_{k \to j}(d_j,p_j,c_j). \label{eq:update-psi}
\end{align}
If the set over which we are minimizing in \eqref{eq:update-b} is empty we
define the minimum to be $\infty$. The messages $b_{j \to i}(d_i,p_i,c_i)$
evolve exactly as they do in
$\eqref{eq:min-sum-update}$.

We will define a collection of messages in terms of $\psi$ and $b$, intended as
an efficient representation of the message passing algorithm, which we
call the A-H form. After defining the messages we will provide the rules by
which these messages are updated. We will also obtain algorithmic complexity
bounds for our representation, which are summarized in Theorem
\ref{thm:complexity}. Then we will demonstrate in Theorem \ref{thm:validity}
the equivalence of the A-H form algorithm to the original description
\eqref{eq:min-sum-update} by showing that iterating the A-H form allows us to
compute the max-marginals \eqref{eq:max-marginals} as if we had done the
iterations in the original form \eqref{eq:min-sum-update}. This theorem will be
proved in two parts, with Lemma \ref{thm:identities} establishing a way to
write all original messages $b_{j \to i}(d_i,p_i,c_i)$ in terms of A-H form
messages and Lemma \ref{thm:updates} establishing how to update the A-H form
messages.

\subsection{Message definitions.}

The messages to be defined fall into three categories based on the type of the
node which sends the message, and the type of the node which receives it. We
first define messages sent from a non-root node $j \in V$ to another non-root
node $i \in V$. Next we will describe messages sent from a non-root node $j \in
V$ to a root node $u \in U$, and finally we describe messages sent from a root
node $u \in U$ to a non-root node $i \in V$.

For $(i,j) \in \EE'$, we define
\[\delta_{i \to j} \defeq \begin{cases} \infty & (i,j) \not\in \EE \\
 0 & (i,j) \in \EE.	
 \end{cases}
\]

Fix $(i,j) \in \EE'$ such that $i,j \in V$.
\begin{itemize}
\item For $3 \le d \le K$, let
\begin{align}
A^{d}_{j \to i} &\defeq \delta_{i \to j} + \min_{c_j \ne *,i} \psi_{j \to i}(d,
i, c_j) \label{eq:var-def-A}\\
&= \delta_{i \to j} + \min_{c_j \ne *,i} \sum_{k \in \partial j \setminus
i}b_{k \to j}(d, i, c_j). \label{eq:var-def-A-b-form}
\end{align}
Note that $-\log \phi_j(d,i,c_j) = 0$ for $c_j \ne *$, so $\psi_{j \to
i}(d,i,c_j) = \sum_{k \in \partial j \setminus i}b_{k \to j}(d,i,c_j)$.
\item For $2 \le d \le K -1$, let
\begin{align*}
B^{d}_{j \to i} &\defeq \delta_{j \to i} + \min_{p_j \ne *,i,\bullet} \psi_{j
\to i}(d, p_j, i)\\
&= \delta_{j \to i} + \min_{p_j \ne *,i,\bullet} \sum_{k \in \partial j
\setminus i}b_{k \to j}(d, p_j, i). 
\end{align*}
\item For $2 \le d \le K$, let
\begin{align}
F^{d}_{j \to i} &\defeq \min_{\substack{p_j \ne *,i,\bullet \\ c_j \ne *,i \\
p_j \ne c_j}}\psi_{j \to i}(d,p_j,c_j) \label{eq:var-def-F}\\
&= \min_{\substack{p_j \ne *,i,\bullet \\ c_j \ne *,i \\p_j \ne c_j}}\sum_{k
\in \partial j \setminus i}b_{k \to j}(d,p_j,c_j). \notag
\end{align}
\item Let
\begin{align}
G_{j \to i} &\defeq \psi_{j \to i}(*,*,*) \label{eq:var-def-G} \\
&= \beta + \sum_{k \in \partial j \setminus i}b_{k \to j}(*,*,*). \notag
\end{align}
\item Let
\begin{equation}\label{eq:var-def-H}
H_{j \to i} \defeq \min\left\{G_{j \to i}, \min_{2 \le d \le K}F^{d}_{j \to
i}\right\}.
\end{equation}
\end{itemize}

Fix $(j,u) \in \EE'$ such that $j \in V$ and $u \in U$.
\begin{itemize}
\item Let
\begin{align*}
A_{j \to u} &\defeq \min_{c_j \ne u, *}\psi_{j \to u}(2,u,c_j) \\
&= \min_{c_j \ne u,*} \sum_{k \in \partial j \setminus u} b_{k \to
j}(2,u,c_j).
\end{align*}
\item For $2 \le d \le K$, let
\begin{align*}
F^{d}_{j \to u} &\defeq \min_{\substack{p_j \ne *,u,\bullet \\ c_j \ne *,u \\
p_j \ne c_j}}\psi_{j \to u}(d,p_j,c_j)\\
&= \min_{\substack{p_j \ne *,u,\bullet \\ c_j \ne *,u \\p_j \ne c_j}}\sum_{k
\in \partial j \setminus u}b_{k \to j}(d,p_j,c_j).
\end{align*}
\item Let
\begin{align*}
G_{j \to u} &\defeq \psi_{j \to u}(*,*,*) \\
&= \beta + \sum_{k \in \partial j \setminus u}b_{k \to j}(*,*,*).
\end{align*}
\item Let
\[H_{j \to u} \defeq \min\left\{G_{j \to u}, \min_{2 \le d \le K}F^{d}_{j \to
u}\right\}.\]
\end{itemize}

Fix $(u,i) \in \EE'$ such that $u \in U$ and $i \in V$.
\begin{itemize}
\item Let
\begin{align*}
B_{u \to i} &\defeq \psi_{u \to i}(1,\bullet,i) \\
&= \sum_{k \in \partial u \setminus i} b_{k \to u}(1,\bullet,i).
\end{align*}
\item Let
\begin{align*}
F_{u \to i} &\defeq \min_{c_u \ne
i,*,\bullet} \psi_{u \to i}(1,\bullet,c_u) \\
&= \min_{c_u \ne i,*,\bullet}\sum_{k \in \partial u \setminus i}b_{k \to
u}(1,\bullet,c_u).
\end{align*} 
\item Let
\begin{align*}
G_{u \to i} &\defeq \psi_{u \to i}(*,*,*) \\
&= \beta + \sum_{k \in \partial u \setminus i}b_{k \to u}(*,*,*).\end{align*} 
\item Let
\begin{align*}
H_{u \to i} &\defeq \min \left\{G_{u \to i}, F_{u \to i}\right\} \\
\end{align*} 
\end{itemize}

Note that when $i \in U$ or $j \in U$ there are different
numbers of messages in each direction. To refer to the set of messages from one
node to another we define for $(i,j) \in \EE'$
\begin{equation}\label{eq:all-messages}
X_{j \to i} = \begin{cases} \left(A^3_{j \to i},\ldots, A^{K}_{j \to i},
B^2_{j \to i}, \ldots, B^{K-1}_{j \to i}, F^2_{j \to i}, \ldots, F^{K}_{j
\to i}, G_{j \to i}, H_{j \to i}\right) & i,j \in V \\
\left(A_{j \to i}, F^2_{j \to i}, \ldots, F^{K}_{j \to i}, G_{j \to i}, H_{j
\to i}\right) & j \in V, i \in U \\
\left(B_{j \to i},F_{j \to i}, G_{j \to i}, H_{j \to i}\right) & j \in U, i \in
V.
\end{cases}
\end{equation}
Note that in all cases the number of messages from $i$ to $j$ is $\bo{K}$.

\subsection{Update rules and the complexity of the A-H form.}

The iteration update rules for these new messages are listed in Tables
\ref{table:V-V}-\ref{table:U-V}. The derivation of these update rules from the
definitions of the messages is given in the proof of Theorem
\ref{thm:validity}. These update rules allow us to provide complexity bounds
for the system:

\begin{table}[p!]
\begin{align*}
A^d_{j \to i} &= \delta_{i \to j} +  \sum_{k \in \partial j \setminus i} H_{k
\to j} + \min\left(0, \min_{k \in \partial^V j \setminus i} \left(A^{d+1}_{k
\to j} - H_{k \to j}\right)\right) \\
A^{K}_{j \to i} &= \delta_{i \to j} +\sum_{k \in \partial j \setminus i} H_{k
\to j} \\
 B^2_{j \to i} &= \delta_{j \to i} + \sum_{k \in \partial j \setminus i} H_{k
\to j} + \min_{w \in \partial^U j}\left(B_{w \to j} - H_{w \to j}\right) \\
B^d_{j \to i} &= \delta_{j \to i} + \sum_{k \in \partial j \setminus i} H_{k
\to j} + \min_{k \in \partial^V j \setminus i}\left(B^{d-1}_{k \to j} - H_{k
\to j}\right) \\
F^2_{j \to i} &= \sum_{k \in \partial j \setminus i} H_{k \to j} + \min_{w \in
\partial^Uj}\left(B_{w \to j} - H_{w \to j}\right) + \min\left(0, \min_{k
\in \partial^Vj \setminus i}\left(A^3_{k \to j} - H_{k \to j}\right)\right) \\
F^d_{j \to i} &= \sum_{k \in \partial j \setminus i} H_{k \to j} + \min_{l \in
\partial^Vj \setminus i}\left(B^{d-1}_{l \to j} - H_{l \to j} + \min\left(0,
\min_{k \in \partial^Vj \setminus \{i,l\}}\left(A^{d+1}_{k \to j} - H_{k \to
j}\right)\right)\right) \\
F^{K}_{j \to i} &= \sum_{k \in \partial j \setminus i} H_{k \to j} + \min_{l
\in \partial^Vj \setminus i}\left(B^{K - 1}_{l \to j} - H_{l \to j} \right)
\\
	G_{j \to i} &= \beta + \sum_{k \in \partial j \setminus i} H_{k \to j} \\
	H_{j \to i} &= \min\left(G_{j \to i}, \min_{2 \le d \le K}F^d_{j \to
i}\right)
\end{align*}
 \caption{Messages from $j \in V$ to $i \in V$, for $(i,j)
\in \EE'$ and $3 \le d \le K - 1$}
\label{table:V-V}
\end{table}

\begin{table}[p!]
\begin{align*}
A_{j \to u} &= \sum_{k \in \partial j \setminus u} H_{k \to j} + \min\left(0,
\min_{k \in \partial^Vj} A^3_{k \to j} - H_{k \to j}\right) \\
F^2_{j \to u} &= \sum_{k \in \partial j \setminus u} H_{k \to j} + \min_{w \in
\partial^Uj \setminus u}\left(B_{w \to j} - H_{w \to j}\right) + \min\left(0,
\min_{k \in \partial^Vj}\left(A^3_{k \to j} - H_{k \to j}\right)\right) \\
F^d_{j \to u} &= \sum_{k \in \partial j \setminus u} H_{k \to j} + \min_{l \in
\partial^Vj}\left(B^{d-1}_{l \to j} - H_{l \to j} + \min\left(0, \min_{k \in
\partial^Vj \setminus l} \left(A^{d+1}_{k \to j} - H_{k \to
j}\right)\right)\right)
\\
F^{K}_{j \to u} &= \sum_{k \in \partial j \setminus u} H_{k \to j} + \min_{l
\in \partial^Vj}\left(B^{K-1}_{l \to j} - H_{l \to j}\right) \\
	G_{j \to u} &= \beta + \sum_{k \in \partial j \setminus u} H_{k \to j} \\
	H_{j \to u} &= \min\left(G_{j \to u}, \min_{2 \le d \le K}F^d_{j \to
u}\right) \\
\end{align*}
 \caption{Messages from $j \in V$ to $u \in U$, for $(u,j)
\in \EE'$ and $3 \le d \le K - 1$}
\label{table:V-U}
\end{table}

\begin{table}[p!]
\begin{align*}
 B_{u \to i} &= \sum_{k \in \partial u \setminus i} H_{k \to u} \\
 F_{u \to i} &= \sum_{k \in \partial u \setminus i} H_{k \to u} + \min_{k \in
\partial u \setminus i}\left(A_{k \to u} - H_{k \to u}\right) \\
 G_{u \to i} &= \beta + \sum_{k \in \partial u \setminus i} H_{k \to u} \\
 H_{u \to i} &= \min\left(G_{u \to i}, F_{u \to i}\right)
\end{align*}
 \caption{Messages from $u \in U$ to $i \in V$, for $(u,i)
\in \EE'$}
\label{table:U-V}
\end{table}
 
\begin{theorem}\label{thm:complexity}
The A-H form has a total of $\bo{n\Delta K}$ messages per iteration. Each
iteration requires a total of $\bo{n\Delta K}$ computations to update the
entire system.
\end{theorem}

\begin{remark}
Recall from \eqref{eq:orig-messages} that the original form
\eqref{eq:min-sum-update} has a total of $\bo{n\Delta^3K}$ messages, while
by Theorem \ref{thm:complexity} the A-H form has $\bo{n \Delta K}$ total
messages.

To perform an iteration of the system, the original form
\eqref{eq:min-sum-update} requires $\bo{n^2\Delta^6K^2}$ computations by
\eqref{eq:orig-computations}, while the A-H form requires only $\bo{n 
\Delta K}$.
\end{remark}

\begin{proof}[Proof of Theorem \ref{thm:complexity}]
For each undirected edge, the A-H form has $\bo{K}$ messages, so the total
number of messages is $\bo{|\EE'|K} = \bo{n\Delta K}$, as desired.

In performing an iteration of the A-H form we proceed node-by-node. As we will
show below, at a node $j$ we will perform $\bo{K\Delta_j}$ total
computations. This implies that we require $\bo{K\sum_{j \in [n]}\Delta_j} =
\bo{n\Delta K}$ total computations for a full iteration of the A-H form, as
claimed.

We now provide more detail on the computations required per node $j$. We 
consider the case $j \in V$; if $j \in U$ the analysis is similar. We will also
focus our attention on the most complex message to compute, which is $F^d_{j
\to i}$ for $3 \le d \le K-1$. We will require $\bo{\Delta_j}$ computations
to compute this for all $i \in \partial j$ and thus $\bo{K\Delta_j}$
computations for all messages sent from $j$ because there are $\bo{K}$
messages from $j \to i$ for each $i \in \partial j$.

To compute $F^d_{j \to i}$ it is necessary to compute two values: 
\[\sum_{k \in \partial j \setminus i} H_{k \to j},\] 
and
\begin{equation}\label{eq:parent-child}
\min_{l \in
\partial^Vj \setminus i}\left(B^{d-1}_{l \to j} - H_{l \to j} + \min\left(0,
\min_{k \in \partial^Vj \setminus \{i,l\}}\left(A^{d+1}_{k \to j} - H_{k \to
j}\right)\right)\right).	
\end{equation}
Rather than compute these directly for each $i$, we first compute
\begin{equation}\label{eq:H-sum}
\sum_{k \in \partial j} H_{k \to j},
\end{equation}
select the three neighbors $l \in \partial^V j$ with the smallest values of
\begin{equation}\label{eq:comp-parent}
 B^{d-1}_{l \to j} - H_{l \to j},
\end{equation}
and select the three neighbors $l \in \partial^V j$ with the smallest values of
\begin{equation}\label{eq:comp-child}
A^{d+1}_{l \to j} - H_{l \to j}.
\end{equation}
We also record the values of \eqref{eq:comp-parent} and \eqref{eq:comp-child}
for the nodes we record.

Once we have these quantities, for each $i \in \partial j$ we can compute
$F^d_{j \to i}$ in $\bo{1}$ computations. To compute $\sum_{k \in \partial j
\setminus i} H_{k \to j}$ we simply subtract $H_{i \to j}$ from
\eqref{eq:H-sum}. To compute \eqref{eq:parent-child} we pick $l$ and $k$
according to the minimum values we have recorded, subject to the constraints
that $l \ne i$, $k \ne i$, and $l \ne k$. These constraints are the reason we
record three  minimizing values, as the first two choices for either $l$ or $k$
may be eliminated by matching $i$ or the other choice. For a given $i
\in \partial j$ we compute $F^d_{j \to i}$ in $\bo{1}$ computations so
computing all the $F^d$ messages sent out from $j$ requires $\bo{\Delta_j}$
computations.

All of the other messages are computed using similar patterns and there are
$\bo{K}$ of them so a total of $\bo{K\Delta_j}$ total computations are required
to compute all the messages sent out from $j$.
\end{proof}

\subsection{Initialization and termination.}
The equations in Tables \ref{table:V-V}-\ref{table:U-V} collectively define
update rules for our algorithm. To initialize the system we set $X^0_{j \to i}
= 1$ for all $(i,j) \in \EE'$, by which we mean we set each element of $X^0_{j
\to i}$ to 1.

Given some a priori bound $T$ on the number of iterations, we then compute
$X^t$ from $X^{t-1}$ for $t \in [T]$ according to the update rules in Tables
\ref{table:V-V}-\ref{table:U-V}. If the messages reach a fixed point with $X^t
= X^{t+1}$ for some $t < T$ we halt the algorithm.

\subsection{Validity of the A-H Form.}
Identities relating original form messages $b_{j \to i}(d_i,p_i,c_i)$ to A-H
form messages for various choices of $d_i$, $p_i$, and $c_i$ are given in
Tables \ref{table:b-V-V}-\ref{table:b-U-V}. It can be checked that for any
choice of $(d_i,p_i,c_i)$ that does not appear in a table we have $b_{j \to
i}(d_i,p_i,c_i) = \infty$ by \eqref{eq:update-b} because no choice of
$(d_j,p_j,c_j)$ can satisfy $h_{ij}(d_i,p_i,c_i,d_j,p_j,c_j) = 1$. The validity
of these identities given the definitions of the A-H form variables is proved
in Lemma \ref{thm:identities} below.

\begin{table}[p!]
\centering
\begin{tabular}{|>{$}c<{$}|>{$}c<{$}|>{$}c<{$}|>{$}c<{$}|}
	\hline
d_i & p_i & c_i & b_{j \to i}(d_i,p_i,c_i) \\
	\hline
* & * & * & H_{j \to i} \\
%* & * & c \in \partial^V i \cup \{\bullet\} & \infty \\
%* & p \in \partial i \cup \{\bullet\} & * & \infty \\
%* & p \in \partial i \cup \{\bullet\} & c \in \partial^V i \cup \{\bullet\} &
%\infty \\
%1 \le d \le K & * & * & \infty \\
%1 \le d \le K & p \in \partial i & * & \infty \\
%1 \le d \le K & * & c \in \partial^V i \cup \{\bullet\} & \infty \\
%1 & p \in \partial i \cup \{*,\bullet\} & c \in \partial^V i \cup
%\{*,\bullet\} & \infty \\
2 & p \in \partial^U i & j & A^3_{j \to i} \\
2 & p \in \partial^U i & c \in \left(\partial^V i \setminus j\right) \cup 
\{\bullet\}& H_{j \to i} \\
%2 & p \in \partial^U i & * & \infty\\
%2 & p \notin \partial^U i & c \in \partial^V i \cup \{*,\bullet\}& \infty \\
3 \le d \le K - 1 & j & c \in \left(\partial^V i \setminus j\right) \cup 
\{\bullet\}& B^{d-1}_{j \to i} \\
3 \le d \le K - 1 & p \in \partial^V i \setminus j & j & A^{d+1}_{j \to i} \\
3 \le d \le K - 1 & p \in \partial^V i \setminus j & c \in \left(\partial^Vi
\setminus j\right) \cup \{\bullet\} & H_{j \to i} \\
K & j & \bullet & B^{K-1}_{j \to i} \\
K & p \in \partial^V i \setminus j & \bullet & H_{j \to i} \\\hline	
\end{tabular}
\caption{Classes of $b_{j \to i}$ messages for $i,j \in V$.}
\label{table:b-V-V}
\end{table}

\begin{table}[p!]
\centering
\begin{tabular}{|>{$}c<{$}|>{$}c<{$}|>{$}c<{$}|>{$}c<{$}|}
	\hline
d_u & p_u & c_u & b_{j \to u}(d_u,p_u,c_u) \\
	\hline
* & * & * & H_{j \to u} \\
1 & \bullet & j & A_{j \to u} \\
1 & \bullet & c \in \left(\partial^V u \setminus j\right) & H_{j \to u} \\
\hline	
\end{tabular}
\caption{Classes of $b_{j \to u}$ messages for $j \in V$, $u \in U$.}
\label{table:b-V-U}
\end{table}

\begin{table}[p!]
\centering
\begin{tabular}{|>{$}c<{$}|>{$}c<{$}|>{$}c<{$}|>{$}c<{$}|}
	\hline
d_i & p_i & c_i & b_{u \to i}(d_i,p_i,c_i) \\
	\hline
* & * & * & H_{j \to i} \\
2 & u & c \in \partial^V i \cup \{\bullet\} & B_{u \to i} \\
2 & p \in \partial^U i \setminus u & c \in \left(\partial^V i\right) \cup
\{\bullet\}& H_{u \to i} \\
3 \le d \le K - 1 & p \in \partial^V i & c \in \left(\partial^Vi
\setminus j\right) \cup \{\bullet\} & H_{u \to i} \\
K & p \in \partial^V i & \bullet & H_{u \to i} \\\hline	
\end{tabular}
\caption{Classes of $b_{u \to i}$ messages for $u \in U$, $i \in V$.}
\label{table:b-U-V}
\end{table}

To use the original form \eqref{eq:min-sum-update} to find
solutions to \eqref{eq:opt-prob} we need to compute the max-marginals
\eqref{eq:max-marginals} for different values of $(d_i,p_i,c_i)$ from the
messages $b_{k \to j}(d_j,p_j,c_j)$. Specifically, we need to find the choices
of $(d_i,p_i,c_i)$ that maximize $\bar{p}_i(d_i,p_i,c_i)$. Thus to show that
we can iterate the A-H form to find the same approximate solutions we would
find by iterating the original form \eqref{eq:min-sum-update} we prove that the
max-marginal can be computed from the A-H form. In particular, we define
\begin{definition}\label{def:rho_i}
For node $i \in [n]$, let $\rho_i$ be any real function of A-H variables
$\{X_{j \to i}\}_{j \in \partial i}$ defined by substituting for $b_{k \to
j}(d_j,p_j,c_j)$ in \eqref{eq:max-marginals} according to the identities
represented by Tables \ref{table:b-V-V}-\ref{table:b-U-V}.	
\end{definition}
We prove the following:
\begin{theorem}\label{thm:validity}
For all $(j,i) \in \EE'$, let the original form be initialized with $b^0_{j \to
i}(d_i,p_i,c_i) = 1$ for all $(d_i,p_i,c_i)$ and let the A-H form be
initialized with $X^0_{j \to i} = 1$. Then for each $i \in [n]$, $d_i
\in [K] \cup \{*\}$, $p_i \in \partial i \cup \{*,\bullet\}$, $c_i \in
\partial^V i \cup \{*,\bullet\}$ and for $0 < t \le T$,
\[ \bar{p}_i(d_i,p_i,c_i) = \rho_i\left(\{X^t_{j \to i}\}_{j \in \partial
i}\right)\]
when both sides are computed with appropriate messages $b^t_{j \to i}$ and
$X^t_{j\to i}$ updated $t$ times according to \eqref{eq:min-sum-update} and
Tables \ref{table:V-V}-\ref{table:U-V}, respectively.
\end{theorem}

We will use two lemmas to prove this theorem:
\begin{lemma}\label{thm:identities}
The identities described in Tables \ref{table:b-V-V}-\ref{table:b-U-V} hold.
\end{lemma}

\begin{proof}
Fix $i,j \in V$.

We focus on two representative messages, namely $A^d_{j \to i}$ for $3 \le d
\le K - 1$ and $H_{j \to i}$. In the first example we will prove
the fifth identity of Table \ref{table:b-V-V} and in the second will prove the 
sixth. The other identities described in Table \ref{table:b-V-V} and all
identities described in Tables \ref{table:b-V-U}-\ref{table:b-U-V} are proved
similarly so we omit the details.

We will first show that the fifth identity of Table \ref{table:b-V-V} holds.
Namely, for $3 \le d \le K - 1$ and $p \in \partial^V i \setminus j$ we will
show
\begin{equation}\label{eq:b-A}
b_{j \to i}(d,p,j) = A^{d+1}_{j \to i}.	
\end{equation}
Plugging $d_i = d$, $p_i = p$, and $c_i = j$ into \eqref{eq:update-b} gives us
\begin{align}
b_{j \to i}(d,p,j) &= \min_{\substack{d_j,p_j,c_j \st \\
h_{ij}(d,p,j,d_j,p_j,c_j) = 1}}\psi_{j \to i}(d_j,p_j,c_j) \label{eq:b-A-first}
\end{align}

Suppose $(i,j) \not\in \EE$. Then $h_{ij}(d,p,j,d_j,p_j,c_j) = 0$ for any
choice of $d_j,p_j,c_j$ and thus $b_{j \to i}(d,p,j) = \infty$. By
\eqref{eq:var-def-A} we have $A^{d+1}_{j \to i} = \infty$, so $b_{j \to
i}(d,p,c) = A^{d+1}_{j \to i}$, as desired.

Now suppose $(i,j) \in \EE$. We now use $h_{ij}(d,p,j,d_j,p_j,c_j) = 1$ to
infer information about $d_j,p_j,c_j$. From \eqref{eq:edge-valid-def-i-parent}
we must have $d_j = d+1$, $p_j = i$, and $c_j \ne i$. Then
\eqref{eq:node-valid-def-non-root} further implies $c_j \ne *$. Substituting
these facts into \eqref{eq:b-A-first} shows
\begin{align*}
b_{j \to i}(d,p,j) &= \min_{c_j \ne *,i}\psi_{j \to i}(d+1,i,c_j).
\end{align*}
Observe that because $\delta_{i \to j} = 0$ the right hand is the definition of
$A^{d+1}_{j \to i}$ as given in \eqref{eq:var-def-A}, so we conclude
\[b_{j \to i}(d,p,j) = A^{d+1}_{j \to i},\]
as desired.

Now we will show the sixth identity of Table \ref{table:b-V-V} holds. Namely,
for $3 \le d \le K - 1$, $p \in \partial^V i \setminus j$, and $c \in
(\partial^V i \setminus j)\cup \{\bullet\}$, we will show 
\begin{align*}
 b_{j \to i}(d,p,c) &= H_{j \to i}.
\end{align*}
Plugging $d_i = d$, $p_i = p$, and $c_i = c$ into \eqref{eq:update-b} gives us
\begin{align*}
b_{j \to i}(d,p,c) &= \min_{\substack{d_j,p_j,c_j \st \\
h_{ij}(d,p,c,d_j,p_j,c_j) = 1}}\psi_{j \to i}(d_j,p_j,c_j).
\end{align*}
We now want to break the minimization into cases based on how we satisfy
$h_{ij}(d,p,c,d_j,p_j,c_j) = 1$. First observe that $p \ne j$ and $c \ne j$
implies we must be in case \eqref{eq:edge-valid-def-neither}, which implies
$p_j \ne i$ and $c_j \ne i$. Then there are two possible cases for how to
satisfy $f_j(d_j,p_j,c_j) = 1$.

First we can choose \eqref{eq:node-valid-def-np} and let $d_j = p_j = c_j = *$.
Otherwise we are in \eqref{eq:node-valid-def-non-root} so we have $d_j = d'$
for some $2 \le d \le K$, $p_j \not\in \{*,\bullet\}$ and $c_j
\not\in \{*,p_j\}$. Combining these cases implies
\begin{align}
b_{j \to i}(d,p,c) &= \min\left\{\psi_{j \to i}(*,*,*),\min_{2 \le d' \le
K}\min_{\substack{p_j \ne *,i,\bullet \\ c_j \ne *,i \\p_j \ne c_j}}\psi_{j
\to i}(d',p_j,c_j)\right\}. \notag
\end{align}
Observe that the right hand side contains the definitions of $G_{j \to i}$ from
\eqref{eq:var-def-G} and $F^{d'}_{j \to i}$ from \eqref{eq:var-def-F}, so we
have
\begin{align}\label{eq:ex-b-H}
b_{j \to i}(d,p,c) &= \min\left\{G_{j \to i}, \min_{2 \le d' \le K}F^{d'}_{j
\to i}\right\} 
= H_{j \to i},
\end{align}
as desired, where the last step is applying the definition of $H_{j \to i}$
from \eqref{eq:var-def-H}.
\end{proof}

\begin{lemma}\label{thm:updates}
The relations in Tables \ref{table:V-V}-\ref{table:U-V} for updating A-H form 
messages are valid.
\end{lemma}

\begin{proof}
We will show how to derive one representative update equation, namely	
\begin{equation}\label{eq:ex-update-A}
A^d_{j \to i} = \delta_{i \to j} + \sum_{k \in \partial j \setminus i}H_{k
\to j} + \min\left\{0, \min_{k \in \partial^V j \setminus i}\left(A^{d+1}_{k
\to j} - H_{k \to j}\right)\right\},
\end{equation}
for $3 \le d \le K -1$, which is the first equation of Table \ref{table:V-V}.
The remainder of Table \ref{table:V-V} is derived via similar methods, as are
the relations in Tables \ref{table:V-U} and \ref{table:U-V}. 

Recall \eqref{eq:var-def-A-b-form}:
\begin{align*}
A^{d}_{j \to i} &= \delta_{i \to j} + \min_{c_j \ne *,i} \sum_{k \in
\partial j \setminus i}b_{k \to j}(d, i, c_j).
\end{align*}
We can split this minimization into two cases, namely $c_j = \bullet$ and $c_j
\in \partial j \setminus i$, which allows us to write
\begin{align*}
A^d_{j \to i} &= \delta_{i \to j} + \min\left\{\sum_{k \in \partial j \setminus
i} b_{k \to j}(d,i,\bullet),\right.  \\
&\hspace{0.2\textwidth} \left.\min_{k\in\partial
j \setminus i} \left(b_{k \to j}(d,i,k) + \sum_{l \in \partial j \setminus
\{i,k\}} b_{l \to j}(d,i,k)\right) \right\}
\end{align*}
We now replace each instance of $b_{k \to j}$ with the appropriate value 
according to the identities in Tables \ref{table:b-V-V} and \ref{table:b-U-V}.
First note that for $k \in \partial^Uj \setminus i$ we have $b_{k \to
j}(d,i,k)= \infty$ because $d \ge 3$ and $k \in U$ implies
$h_{jk}(d,i,k,d_k,p_k,c_k) = 0$ since $g_{jk}(d,i,k,d_k,p_k,c_k) = 1$ if and
only if $d_k = d + 1$ and if $d_k = d + 1 > 1$ then $f_k(d_k,p_k,c_k) = 0$.
Thus minimizing over $k \in \partial j \setminus i$ is equivalent to minimizing
over $k \in \partial^V j \setminus i$ so we can write
\begin{align*}
A^{d}_{j \to i}  &= \delta_{i \to j} + \min\left\{\sum_{k \in \partial j
\setminus i} b_{k \to j}(d,i,\bullet),\right. \\
&\hspace{0.2\textwidth}\left. \min_{k\in\partial^V
j \setminus i} \left(b_{k \to j}(d,i,k) + \sum_{l \in \partial j \setminus
\{i,k\}} b_{l \to j}(d,i,k)\right) \right\}.
\end{align*}
We now replace these particular messages via $b_{k \to j}(d,i,k) = A^{d+1}_{k
\to j}$ from \eqref{eq:b-A} (which is the fifth line of Table
\ref{table:b-V-V}) and $b_{l \to j}(d,i,k) = H_{l \to j}$ and $b_{k
\to j}(d,i,\bullet) = H_{k \to j}$ from \eqref{eq:ex-b-H} (which is the
sixth line of Table \ref{table:b-V-V}), so we have
\begin{align*}
A^d_{j \to i} &= \delta_{i \to j} + \min\left\{\sum_{k \in \partial j
\setminus i} H_{k \to j}, \min_{k \in \partial j \setminus i} \left(A^{d+1}_{k
\to j} + \sum_{l \in \partial j \setminus \{i,k\}} H_{l \to j} \right)\right\}
\\
&= \delta_{i \to j} + \sum_{k \in \partial j \setminus i}H_{k \to j} +
\min\left\{0, \min_{k \in \partial j \setminus i}\left(A^{d+1}_{k \to j} - H_{k
\to j}\right)\right\},
\end{align*}
where the second line simply extracts the sum from the minimization. This
equation now matches \eqref{eq:ex-update-A}.
\end{proof}

We are now prepared to prove Theorem \ref{thm:validity}.

\begin{proof}[Proof of Theorem \ref{thm:validity}]
Tables \ref{table:b-V-V}-\ref{table:b-U-V} directly provide the functions
$\rho_i$ required to compute max-marginals. This substitution is valid after
any number of iterations because it is valid from the original definitions of
the A-H form variables by Lemma \ref{thm:identities} and Lemma
\ref{thm:updates} proves the A-H variables are iterated by rules that follow
directly from the iteration rules for the original form variables.
\end{proof}

\section{Constructing solutions.}\label{sec:reconstruction}
Because we have no guarantee that messages will converge, or that if they
converge they will converge to an optimal solution, we elect to compute a 
feasible solution at every iteration of the algorithm and return the best
solution found over all iterations. In particular we run our algorithm for a
fixed number of iterations unless we converge to a fixed point before this
limit. Finding all these solutions comes at an additional computational cost,
but it gives consistently better solutions than simply using the result at the
end of the iterations.

To initialize our solution construction algorithm, let $\mathcal{A} = V$ be
the set of non-root nodes which we have not yet assigned to a path. Nodes will
be removed from this set as we determine their role in our constructed
solution. For each node $u \in U$, we will construct a (possibly empty) path
$P_u$, and our full solution will be the collection of these paths.  The first
step is to compute the max-marginals for the possible choices at node $u$,
which are either ``do not participate in a path'' (in which case $P_u =
\emptyset$) or ``start a path by choosing a child from among available
neighbors of $u$''. We want to maximize the max-marginal among these choices. 
Since our variables are in Min-Sum form we actually compute the negative log of
the max-marginals and choose the minimizing configuration. To choose a
configuration at the node $u$, we use the functions $\rho_i$ from Definition
\ref{def:rho_i} to compute
\begin{align}
-\log \bar{p}_u(*,*,*) &= \sum_{k \in \partial u}H_{k \to u} + \beta
\label{eq:max-marg-root-dnp}\\
-\log \bar{p}_u(1,\bullet,j) &= \sum_{k \in \partial u}H_{k \to u} + A_{j \to
u} - H_{j \to u}, \quad j \in \partial^V u \cap \mathcal{A}.
\label{eq:max-marg-root-start}
\end{align}
and find the minimum value. If the minimum is \eqref{eq:max-marg-root-dnp},
we let $P_u = \emptyset$ and proceed to the next root node. If the minimum is
\eqref{eq:max-marg-root-start} for some neighbor $j$, we initialize our new
path as $P_u = \{u,j\}$ and remove $j$ from the set of available nodes by
setting $\mathcal{A} = \mathcal{A} \setminus \{j\}$. Note that
\eqref{eq:max-marg-root-start} considers only neighbors which are also elements
of $\mathcal{A}$ insuring that $P_u$ does not include any nodes used in a
previous path. Further note that we must have $(u,j) \in \EE$ so this is a
valid start to a path because $j \in \partial^V u$ and $\EE$ contains no edges
in $V \times U$. We next proceed to the node $j$ to determine whether $P_u$
should be extended or added to the solution as is.

This process is essentially the same as that for $u$, involving choosing the
optimal value of the max-marginal at $j$, but we are now constrained to
configurations with parent $u$. This means that we are choosing whether to end
the path at $j$ or continue to one of $j$'s neighbors in $V$, so we compute
\begin{align}
	-\log \bar{p}_j(2,u,\bullet) &= \sum_{k \in \partial j} H_{k \to j} + B_{u \to
j} - H_{u \to j} \label{eq:max-marg-top-end} \\
	-\log \bar{p}_j(2,u,i) &= \sum_{k \in \partial j} H_{k \to j} + B_{u \to j} -
H_{u \to j} + A^3_{i \to j} - H_{i \to j}, \quad i \in \partial^Vj
\cap \mathcal{A}\label{eq:max-marg-top-cont}
\end{align}
and choose the minimum. If the minimum is \eqref{eq:max-marg-top-end} the path
is terminated so we add $P_u$ to our solution as is and proceed to the next
root node. Otherwise the minimum is achieved by a particular choice $i$, so we
set $P_u = \{u,j,i\}$, remove $i$ from the set of available nodes by setting
$\mathcal{A} = \mathcal{A} \setminus \{i\}$, and proceed to determine whether
$P_u$ should be extended or added to the solution as is. Note again that $i
\in \mathcal{A}$ guarantees we have not previously selected $i$ to participate
in any other path. To confirm that $(j,i) \in \EE$, note that the first line of
Table \ref{table:V-V} implies $A^3_{i \to j} = \infty$ if $(j,i) \not\in \EE$,
in which case \eqref{eq:max-marg-top-cont} is infinite for that choice of $i$.
Such an $i$ cannot be the chosen minimum value because
\eqref{eq:max-marg-top-end} is always finite as can be seen from the last two
lines of Table \ref{table:V-V} and the first line of Table \ref{table:b-U-V}.
Thus if we choose $i$ to continue the path we must have $(j,i) \in \EE$ which
guarantees this is a valid choice.

After the first two nodes of a path the process becomes more generic, so let
us suppose we have some $P_u$ such that $|P_u| = d \ge 3$ and the last two
nodes of the path $P_u$ are are $p$ and $j$. We still need to determine whether
to extend the path $P_u$ or add it to the solution as is. If $d = K$ then we
know that we cannot extend the path, so we add it to the solution and move on
to the next root node. Otherwise, we compute
\begin{align}
	-\log \bar{p}_j(d,p,\bullet) &= \sum_{k \in \partial j} H_{k \to j} +
B^{d-1}_{p \to j} - H_{p \to j} \label{eq:max-marg-mid-end} \\
	-\log \bar{p}_j(d,p,i) &= \sum_{k \in \partial j} H_{k \to j} + B^{d-1}_{p \to
j} - H_{p \to j} + A^{d+1}_{i \to j} - H_{i \to j}, \quad i \in \partial^Vj
\cap \mathcal{A} \label{eq:max-marg-mid-cont}
\end{align}
and choose the minimum. If the minimum is \eqref{eq:max-marg-mid-end} we leave
$P_u$ as is, add it to the solution and move on to the next root node. If the 
minimum is \eqref{eq:max-marg-mid-cont} for some $i$, then we append $i$ to the
end of $P_u$ and remove $i$ from $\mathcal{A}$. We now have a path of length $d
+1$ and can repeat the above procedure until we terminate the path. As above,
we restrict our attention to nodes in $\mathcal{A}$ to insure no node is
assigned to multiple paths. Furthermore, for $i \in \partial^Vj \cap
\mathcal{A}$ if $(j,i) \not\in \EE$ then $A^{d+1}_{i \to j} = \infty$ but
\eqref{eq:max-marg-mid-end} is always finite because $H_{k \to j}$ is finite
for any pair of nodes $k$ and $j$ and $B_{p \to j}^{d-1}$ is finite because we
are guaranteed to have $(p,j) \in \EE$ by the fact that these are the last two
nodes of the partially constructed path $P_u$. Thus if we choose to continue
the path we do so along an edge with the correct direction to a node that is
not participating in another already constructed path, and so the end result of
this process is a feasible collection of node-disjoint directed paths.

Notice, however, that the result depends on the order in which we process
roots, so in practice we often repeat the procedure several times with
different orders and choose the best solution among those we find.

\section{Alternative algorithms.}\label{sec:algorithms}
Our numerical results in section \ref{sec:performance} below will compare our
efficient BP algorithm to three alternatives: a Greedy algorithm, an IP based
algorithm for KEP from \cite{Anderson20012015}, and a novel IP based algorithm.
We now describe these algorithms.

We propose the Greedy algorithm as a simple to implement approximation
algorithm with good scaling behavior. It is greedy in the sense of adding paths
to the solution one at a time and adding the longest possible path at each
step, but when adding a path the search for the longest path is exhaustive.
Given an arbitrary ordering of root notes, the Greedy algorithm proceeds as
follows. Let $\mathcal{A} = V$ be the set of non-root nodes that have not yet
been added to a path. From a root node $u \in U$, all possible directed
paths with length at most $K$ consisting of $u$ and a collection of nodes in
$\mathcal{A}$ are explored and the longest such path is chosen, with ties being
broken arbitrarily. The longest possible path is found by recursively looking
for the longest path from each neighbor, keeping track of the depth bound to
insure a path longer than $K$ is not chosen. The chosen path is then added to
the solution and its nodes are removed from $\mathcal{A}$. If the path consists
only of the node $u$ it is discarded. The algorithm now proceeds to the next
root node in the given order.
 
This algorithm is guaranteed to generate a feasible collection of
directed paths, but the size of that collection is highly dependent on the
order in which we consider root nodes, so we repeat the process for a large
number of different orderings of the root nodes, choosing the best solution.
For the numerical results to follow we often used about 200 randomly generated
orders. In most cases, the number of orderings is chosen so the running time of
the Greedy algorithm is roughly equal to the message passing algorithm (which
also includes a similar randomization step, as noted in section
\ref{sec:reconstruction}).

The KEP algorithm from \cite{Anderson20012015} is based on a Traveling Salesman
IP formulation approach to the kidney exchange problem and can find bounded
length cycles as well as paths, which may be bounded or unbounded in length.
For our comparisons we enforce our path length bound $K$ and disallow all
cycles.

Finally we have implemented a simple IP formulation of the problem based on the
same description we used in the development of the message passing algorithm.
We call this the Parent-Child-Depth (PCD) formulation. We let
\[\partial^+ i \defeq \{j \in \partial i : (i,j) \in \EE\}\]
and
\[\partial^- i \defeq \{j \in \partial i : (j,i) \in \EE\}.\]
The formulation is as follows:
\begin{align}
\text{maximize}\quad &\sum_{i \in [n]}x_i \label{eq:ip-obj}\\
\text{s.t.}\quad x_i &= \sum_{j \in \partial^- i}p_{ji} + p_{\bullet i}
&\forall i \in [n] \label{eq:ip-x}\\
\sum_{j \in \partial^+ i}c_{ij} &\le x_i &\forall i \in [n]
\label{eq:ip-only-one} \\
p_{\bullet i} &\le \sum_{j \in \partial^+ i}c_{ij} &\forall i \in [n]
\label{eq:ip-root-child} \\
d_i &= d_{\bullet i} + \sum_{j \in \partial^- i}d_{ji} &\forall i \in [n]
\label{eq:ip-d} \\
0 \le d_i &\le K &\forall i \in [n]\label{eq:ip-d-bounds} \\
d_{\bullet i} &= p_{\bullet i} &\forall i \in [n]\label{eq:ip-ddot} \\
p_{\bullet i} &= 0 &\forall i \in V \label{eq:ip-root} \\
p_{ji} &= c_{ji} &\forall (j,i) \in \EE \label{eq:ip-pc-match}\\ 
d_{ji} &= p_{ji}(d_j + 1) &\forall (j,i) \in \EE\label{eq:ip-depth-cons} \\
p_{ji} + c_{ij} &\le 1 &\forall i \in [n], j \in \partial^-i\cap \partial^+ i
 \label{eq:ip-d-no-2-cycle} \\
x_i,p_{\bullet i} &\in \{0,1\}, \quad d_i, d_{\bullet i} \in \Z & \forall i
\in [n] \label{eq:ip-var-lone}\\
p_{ji} &\in \{0,1\}, \quad d_{ji} \in \Z &\forall i \in [n], j \in \partial^-
i \label{eq:ip-var-par} \\
c_{ij} &\in \{0,1\} &\forall i \in [n], j \in
\partial^+ i  \label{eq:ip-var-child}
\end{align}

\begin{prop}\label{prop:PCD-valid}
There is a bijection between feasible solutions to the IP problem
\eqref{eq:ip-obj}-\eqref{eq:ip-var-child} and the set $M$ defined in
\eqref{eq:feasible-set-def} and an optimal solution to the IP problem maximizes
the objective function $H(d,p,c)$ defined in \eqref{eq:graph-cost-def}.
\end{prop}

\begin{remark}
This proposition states that the IP problem
\eqref{eq:ip-obj}-\eqref{eq:ip-var-child} is in fact the same problem we deal
with elsewhere in this paper. Before providing a formal proof, we provide some
intuition for the formulation. The variable $x_i$ acts as an indicator for
whether the node $i$ participates in a path, with the various $p_{ji}$ and
$p_{\bullet i}$ indicators representing each possible choice of parent for node
$i$ (including the possibility of starting a path). Similarly the $c_{ij}$
variables represent the possible choices of child for the node $i$. The integer
variable $d_i$ represents the depth of $i$ in the path. This depth variable is 
decomposed into variables $d_{ji}$ and $d_{\bullet i}$ which will be equal to
$d_i$ if the parent of $i$ is $j$ or $i$ starts a path, respectively, and zero
otherwise. Thus at most one of $d_{\bullet i}$ and the $d_{ji}$ will be
nonzero, representing the actual choice of parent implied by the values of
$p_{\bullet i}$ and the $p_{ji}$. The various constraints serve to guarantee
that this interpretation of these variables is accurate and all feasibility
constraints for the feasible set $M$.

More specifically, constraint \eqref{eq:ip-x} insures that a node
participates if and only if it has a parent or starts a path. The remaining 
constraints insure that each node has at most one parent and one child, and
that all variables agree locally with each other. Note
\eqref{eq:ip-root-child}, which requires a root to have a child (so each path
contains at least one edge). The depth bound is enforced by requiring $d_i \le
K$ and insuring that $d_i$ accurately represents the depth of $i$. Note in
particular \eqref{eq:ip-depth-cons}, which insures that depth increases from
parent to child. This constraint is quadratic.
\end{remark}

\begin{proof}[Proof of Proposition \ref{prop:PCD-valid}]
Given $(d',p',c') \in M$, we can construct IP variables $(x,d,p,c)$
node-by-node. For each $i \in [n]$ there are three possible cases, represented
by \eqref{eq:node-valid-def-np}-\eqref{eq:node-valid-def-non-root} because
$f_i(d_i',p_i',c_i') = 1$ is guaranteed by $(d',p',c') \in M$:
\begin{itemize}	
\item If $d_i' = p_i' = c_i' = *$, then the node $i$ does not participate in a
path, so we let $x_i = p_{\bullet i} = d_{i\bullet} = 0$, for $j
\in \partial^- i$ we let $p_{ji} = d_{ji} = 0$, and for $j \in \partial^+ i$ we
let $c_{ij} = 0$.
\item If $p_i' = \bullet$, then the node $i$ starts a path, so we let
$x_i = p_{\bullet i} = 1$, $d_i = d_{\bullet i} = 1$, and for each $j \in
\partial^- i$ we let $p_{ij} = d_{ij} = 0$. By \eqref{eq:node-valid-def-root}
we know $i$ has a child $c_i' \in \partial^+ i$, so we can let $c_{ic_i'} = 1$
and for $j \in \partial^+ i \setminus c_i'$ we let $c_{ij} = 0$. Note that we
must also have $d_i' = 1$, so we have $d_i = d_i'$.
\item If $p_i' \not\in \{*,\bullet\}$ then the node $i$ has a parent $p_i'
\in \partial^- i$, so we let $x_i = p_{p_i'i} = 1$ and let $p_{\bullet i} = 0$
and for $j \in \partial^- i \setminus p_i'$ let $p_{ji} = 0$. By
\eqref{eq:node-valid-def-non-root} we have $2 \le d_i' \le K$, so we let $d_i
= d_{p_i'i} = d_i'$ and let $d_{\bullet i} = 0$ and for $j \in \partial^- i
\setminus p_i'$ let $d_{ji} = 0$. If $c_i' = \bullet$ we let $c_{ij} = 0$ for
all $j \in \partial^+ i$ because $i$ has no child. Otherwise we let $c_{ic_i'}
= 1$ and for all $j \in \partial^+ i \setminus c_i'$ let $c_{ij} = 0$.
\end{itemize}
With variables constructed in this way it is straightforward to check that all
of the constraints are satisfied, so we omit the details.

Now suppose we have a feasible set of IP variables $(x,d,p,c)$. We will 
construct $(d',p',c')$ as described in Section \ref{sec:setup} and show that it
is an element of the set $M$. We proceed node-by-node. For $i \in [n]$:
\begin{itemize}
\item If $x_i = 0$, let $d_i' = p_i' = c_i' = *$. 
\item If $x_i = 1$ and $p_{\bullet i} = 1$, let $d_i' = 1$, $p_i' = \bullet$. 
By \eqref{eq:ip-root-child} there is some $j \in \partial^+ i$ such that
$c_{ij} = 1$. Let $c_i' = j$.
\item If $x_i = 1$ and $p_{\bullet i} = 0$, by \eqref{eq:ip-x} there exists
some $j \in \partial^- i$ such that $p_{ji} = 1$. Let $p_i' = j$, and let $d_i'
= d_i$. If $c_{ik} = 0$ for all $k \in \partial^+ i$, let $c_i' = \bullet$.
Otherwise \eqref{eq:ip-only-one} implies there is exactly one $k
\in \partial^+ i$ such that $c_{ik} = 1$, so we let $c_i' = k$.
\end{itemize}
We will now show $(d',p',c') \in M$. 

For $i \in [n]$ the three cases above correspond exactly to the cases
\eqref{eq:node-valid-def-np}-\eqref{eq:node-valid-def-non-root} to guarantee
$f_i(d_i',p_i',c_i') = 1$. Note in particular that \eqref{eq:ip-root}
guarantees $i \in U$ for any $i$ with $p_i'=\bullet$. If $p_i'
\not\in \{*,\bullet\}$ then we have $i \in V$ because $p_i' = j$ if and only if
$p_{ji} = 1$ for some $j \in \partial^- i$ but by assumption $\partial^- i =
\emptyset$ for $i \in U$.

For $i \in [n], j \in \partial i$, we have three cases. If $p_i'=j$ then we
know $p_{ji} = 1$ and $j \in \partial^- i$, so \eqref{eq:ip-pc-match} implies
$c_{ji} = 1$ and thus $c_j' = 1$. \eqref{eq:ip-d-no-2-cycle} guarantees $p_j'
\ne i$ and $c_i' \ne j$. $j \in \partial^- i$ guarantees $(j,i) \in \EE$ and
\eqref{eq:ip-depth-cons} guarantees $d_i' = d_j' + 1$. The case $p_j' = i$ is
similar. If $p_i' \ne j$ and $p_j' \ne i$ \eqref{eq:ip-pc-match} guarantees
$c_i' \ne j$ and $c_j' \ne i$. Thus we are always in one of the cases
\eqref{eq:edge-valid-def-j-parent}-\eqref{eq:edge-valid-def-neither}, so
$g_{ij}(d_i',p_i',c_i',d_j',p_j',c_j') = 1$.

Thus we conclude $(d',p',c') \in M$, so we have established a bijection between
the feasible set of the IP and the set $M$.

To see that the optimization is the same, we need only note that $x_i = 1$ if
and only if $p_i' \ne *$. Thus by \eqref{eq:node-cost-def} we have $x_i =
\eta_i(p_i',c_i',d_i')$ for all $i \in [n]$ and so by \eqref{eq:graph-cost-def}
we have $\sum_{i \in [n]}x_i = \sum_{i \in [n]} \eta(d_i',p_i',c_i') =
H(d,p,c)$. In other words, both are simply measuring the number of nodes that
participate in a path.
\end{proof}

Note that this formulation is polynomial in $n$, both in the number of
variables and in the number of constraints, and is therefore simple to
implement. As we will see below, for a number of instances, with $n = 1000$ and
$K = 15$, PCD produces optimal solutions in more than 50\% of tested graphs
(see Table \ref{fig:long-random}), while KEP does not find optimal solutions in
the same allowed time.

\section{Numerical results.}\label{sec:performance}
We compare the performance of our algorithm, which we call BP, to the three
alternative algorithms described in section \ref{sec:algorithms}. Note that the
two IP-based algorithms (KEP and PCD) allow us to solve the
problem to optimality given sufficient time, whereas the Greedy and BP 
algorithms have no optimality guarantees.

When running the BP algorithm we use the value $\beta = 0.01$ for the
parameter defined in \eqref{eq:node-potential}. 
Informal testing of different values of $\beta$ indicated that this value had,
in some cases, a small advantage of about 1\% in the size of solutions over
larger $\beta$ values while smaller $\beta$ values resulted in significantly
worse solutions. When this advantage was not present the value of $\beta$ had
no noticeable effect on the solution size.

For all cases except the three largest real networks we use 200 orders of root
nodes for the Greedy algorithm and 5 orders of root nodes for each iteration of
BP. For the largest graphs these are decreased to 50 and 2, respectively. In
all instances BP is run for a maximum of $T=50$ iterations, terminating earlier
only if it has converged to a fixed point or exceeded a preset time limit. In
most cases BP did not consistently converge, though it did converge somewhat
more often for smaller graphs with $n = 1000$ and for longer allowed paths
with $K = 15$. We have not seen a significant correlation between instances
where BP converges and those where it finds significantly better solutions, but
convergence does improve running time simply by allowing the algorithm to
terminate earlier.

\subsection{Random graph comparisons.}
We report our results on random graphs constructed as follows. We let each
ordered pair of nodes in $V \times V$ and $U \times V$ have an edge
independently with probability $p$. We set $p=c/n$ and vary $c$. Thus for every
pair $(i,j) \in U \cup V \times V$ the directed edge $(i,j)$ is present with
probability $c/n$ and absent with probability $1 - c/n$. 

\begin{table}
	\centering
	\begin{tabular}{|c|c|c c c|c c c|}
		\multicolumn{2}{c}{} & \multicolumn{3}{c}{Nodes in solution} &
\multicolumn{3}{c}{Running Time (s)} \\
		\hline
		Root\% & $c$ & KEP & Greedy & BP & KEP & Greedy & BP \\
		\hline
10 & 2 & \cellcolor{optimal} 384.9 & 364.4 & \cellcolor{better} 376.6 & 3.0 &
0.6 & 2.2 \\
10 & 3 & \cellcolor{optimal} 460.8 & 446.4 & \cellcolor{better} 457.8 & 16.2 &
0.9 & 4.5 \\
10 & 4 & 426.4 & 480.4 & 469.2 & 63.5 & 0.9 & 4.0 \\
20 & 2 & \cellcolor{optimal} 641.6 & 554.6 & \cellcolor{better} 603.4 & 3.3
& 0.9 & 2.6 \\
20 & 3 & 813.6 & 685.9 & \cellcolor{better} 746.3 & 63.5 & 1.3 & 3.8 \\
20 & 4 & 721.8 & 763.1 & \cellcolor{best} 769.0 & 63.3 & 1.8 & 2.1 \\
25 & 2 & \cellcolor{optimal} 700.7 & 600.3 & \cellcolor{better} 654.5 & 2.9 &
1.0 & 2.4 \\
25 & 3 & 893.7 & 726.6 & \cellcolor{better} 787.0 & 56.4 & 1.7 & 3.9 \\
25 & 4 & 818.7 & 801.3 & \cellcolor{best} 834.9 & 63.9 & 2.3 & 2.5 \\
		\hline
	\end{tabular}
	\caption{Random graphs on $n = 1000$ nodes with maximum path length $K = 5$
and edge probability $c/n$. 100 samples were run for each row, with $\beta =
0.01$ for BP. \colorbox{optimal}{Dark shaded KEP results} are cases where KEP
was optimal in at least 99\% of samples. BP results are labeled by
\colorbox{better}{finding more nodes than Greedy} and \colorbox{best}{finding
the most nodes of any method}.}
	\label{fig:random} 
\end{table}

Table 1 shows results for some varieties of random graphs. The column
``Root\%'' indicates the percentage of nodes which are in the set $U$, namely
$|U|/n$. The column ``$c$'' indicates the value of the parameter $c$ for the
graphs tested. The remaining columns summarize the results of testing the
algorithms on 100 random graphs, first showing the average number of nodes in
the returned solution and then showing the average running time in seconds. For
the IP based algorithms a time limit of 60 seconds was implemented by the IP 
solver and if the optimal solution was not found within that time the best
feasible solution identified by the solver was used. While in many cases KEP
returned a nontrivial feasible solution, if PCD did not find the optimal
solution it almost always returned the trivial feasible solution zero. For the
cases represented in Table \ref{fig:random} PCD never returned a nontrivial
feasible solution so we omit PCD from the table. It does, however, appear in
other cases below and performs well in some ranges outside those that appear in
tables, as will be discussed later.

Testing over a wide range of parameter values we find that in many instances
one or both of the IP algorithms runs to termination in time comparable to BP
or Greedy. In these cases Greedy and BP both generally return suboptimal
solutions, with both algorithms having ranges of parameters where they enjoy a
small advantage over the other. This advantage ranges from Greedy finding about
5\% more nodes than BP to BP finding around 10\% more nodes than Greedy. Table
\ref{fig:random} shows results for random graphs with $1000$ nodes and a
variety of percentages of root nodes. The algorithms were run with path length
bound $K = 5$ and edge probability $c/n$ for $c = 2,3,4$. For $c = 2$ KEP is
optimal in every sample and has a running time comparable to BP. BP finds
feasible solutions with 93\%-98\% of the nodes of the optimal solution, whereas
the Greedy solutions have 85\%-95\%. As we increase $c$, KEP finds fewer
optimal solutions, though it is still optimal in all but 1 of our 100 samples
when 10\% of the nodes of the graph are roots for $c = 3$. In most cases BP
finds more nodes than Greedy does while falling short of even the suboptimal
KEP solutions, but for $c = 4$ with 20\% and 25\% root nodes, BP performs the
best among the compared algorithms, finding the most nodes in time comparable
to Greedy and significantly shorter than KEP.

While we do not show full results here, we can provide some sense of the
behavior of the algorithms for graphs just outside the regimes in Table
\ref{fig:random}. If the root percentage is decreased below 10\% KEP is
generally optimal and both BP and Greedy often find this optimal solution,
likely because the scarcity of roots leads to less interaction between paths in
feasible solutions. If the root percentage is increased above 25\% BP still
finds more nodes than Greedy but significantly fewer than KEP, which generally 
finds optimal solutions. Though PCD did not find nontrivial solutions for the
cases in Table \ref{fig:random}, it is often optimal for these cases with
many roots. If the root percentage is kept in the 10\%-25\% range but $c$ is
decreased to $1$, both KEP and PCD are always optimal, with Greedy missing a
small percentage of nodes and BP slightly behind Greedy. If $c$ is increased
above $4$, both Greedy and BP find significantly (10\%-20\%) more nodes than
KEP, which is never optimal, and Greedy usually has a small (1\%-3\%) advantage
over BP. In this regime PCD finds no nontrivial solutions.

\begin{table}
	\centering
	\begin{tabular}{|c| c c c |c c c|}
		\multicolumn{1}{c}{} & \multicolumn{3}{c}{Nodes in solution} &
\multicolumn{3}{c}{Running Time (s)} \\
		\hline
		$c$ & KEP  & Greedy & BP & KEP  & Greedy & BP \\
		\hline
2 & \cellcolor{optimal} 6426.0 & 5371.5 & \cellcolor{better} 5973.7	 & 150.3
& 39.3 & 42.6 \\
3 & 6792.8 & 6640.2 & \cellcolor{best} 7332.5 & 1229.4 & 60.0 & 65.8 \\
4 & 7199.9 & 7413.1 & \cellcolor{best} 7492.0 & 1250.4 & 87.3 & 48.6 \\
		\hline
	\end{tabular}
	\caption{Random graphs on $n = 10000$ nodes with 20\% root nodes, maximum path
length $K = 5$, and edge probability $c/n$. 100 samples were run for each
row, with $\beta = 0.01$ for BP. \colorbox{optimal}{Dark shaded KEP results}
were optimal in all samples. BP results are labeled by
\colorbox{better}{finding more nodes than Greedy} and \colorbox{best}{finding
the most nodes of any method}.}
	\label{fig:big-random} 
\end{table}

Table \ref{fig:big-random} shows tests over some of the same parameters as
Table \ref{fig:random} on larger graphs, with 10,000 nodes. For these larger
graphs the time limit for KEP and PCD is increased to 20 minutes, and again PCD
is omitted because it never finds nontrivial solutions within the time limit.
Most general trends remain the same, with KEP always finding optimal
solutions within the time limit for $c = 2$ but BP performing better for $c =
3$ and $c = 4$. Specifically, BP now has an advantage over KEP for $c = 3$.
Another factor to note is that BP has a significant advantage over Greedy in
running time in the $c = 4$ case, which we can credit to the fact that BP
converged in all samples for that case, running an average of 19 iterations as
compared to 50 for the other cases.

\begin{table}
	\centering
	\begin{tabular}{|c|c| c c c c|c c c c|}
		\multicolumn{2}{c}{} & \multicolumn{4}{c}{Nodes in solution} &
\multicolumn{4}{c}{Running Time (s)} \\
		\hline
		$K$ & $c$ & KEP & PCD & Greedy & BP & KEP & PCD & Greedy & BP \\
		\hline
10 & 2 & 656.1 & 418.3 & 571.3 & \cellcolor{best} 671.7	& 63.0 & 59.9 & 1.4 &
3.7 \\
10 & 3 & 784.3 & 9.0 & 715.9 & \cellcolor{better} 758.2	& 67.8 & 60.2 & 1.8 &
3.3 \\
10 & 4 & 830.2 & 0 & 796.5 & \cellcolor{better} 812.6	& 69.1 & 60.2 & 2.3 &
2.6 \\
10 & 5 & 841.2 & 0 & 845.5 & \cellcolor{best} 852.6 & 72.6 & 60.2 & 2.8 & 2.7
\\
15 & 2 & 686.3 & \cellcolor{optimal} 709.9 & 584.9 & \cellcolor{better} 684.7
& 63.1 & 10.0 & 1.3 & 3.6 \\
15 & 3 & 868.0 & 779.7 & 728.6 & \cellcolor{better} 789.7 & 68.4 & 34.1 & 2.0 &
3.7 \\
15 & 4 & 921.0 & 760.2 & 806.6 & \cellcolor{better} 835.6 & 69.3 & 40.1 & 2.7 &
3.0 \\
15 & 5 & 938.0 & 693.8 & 855.0 & \cellcolor{better} 869.6 & 72.2 & 47.4 & 3.2 &
3.1 \\
		\hline
	\end{tabular}
	\caption{Random graphs on $n = 1000$ nodes with 15\% root nodes, maximum path
length $K$, and edge probability $c/n$. 100 samples were run for each row,
with $\beta = 0.01$ for BP. \colorbox{optimal}{Shaded PCD results} were optimal
in 96\% of samples, and in the other $K = 15$ cases PCD was optimal in at
least 50\% of samples. BP results are labeled by \colorbox{better}{finding more
nodes than Greedy} and \colorbox{best}{finding the most nodes of any method}.}
	\label{fig:long-random} 
\end{table}

In Table \ref{fig:long-random} we provide summary results for problems with
longer allowed paths; that is with $K = 10$ and $K = 15$. Here KEP is
almost never optimal, though it does often provide a feasible solution with
more nodes than BP or Greedy, albeit at the cost of running to the
cutoff time rather than terminating in a few seconds. In all cases for this
regime BP finds more solutions than Greedy while running in comparable time,
and in two of these cases it also on average finds more nodes than any other
tested algorithm. In some cases the solution size advantage over Greedy is
significant, with BP solutions more than 10\% larger than Greedy solutions.
Note that PCD performs rather well for $K = 15$, often
finding the optimal solution, though the average size of the solution is
reduced by the fact that PCD rarely returns a nontrivial feasible solution,
generally either returning the optimal solution or an empty solution. PCD found
the optimal solution in at least 50\% of samples for all the $K = 15$ cases,
while KEP only found any for $c = 2$, for which it was optimal in only 6\% of
cases. Thus there is an advantage to PCD over KEP for all $K = 15$ cases as
it much more regularly finds the optimal solution. When it does find the
optimal solution it is also relatively fast, with an average runtime of 23
seconds for cases across all 4 values of $c$ in which PCD found the optimal
solution. BP and Greedy both remain significantly faster, and BP does on
average find better solutions than either PCD or Greedy.
%NOTE: that 23 seconds is calculated directly from the individual graph records

For similar graphs just outside the regime presented in Table
\ref{fig:long-random} we have a couple situations. As in the case of a shorter
path bound, decreasing $c$ to 1 makes KEP and PCD optimal in all cases, with
Greedy missing a small percentage of nodes and BP slightly behind Greedy.
Increasing $c$ above $5$ leads to Greedy and BP having essentially identical
performance, both generally better than KEP and PCD. PCD still finds optimal
solutions in a reasonable percentage of $K = 15$ cases, however. If the path
length bound is kept as long as 15 or increased up to 25, PCD performs very
well over a wide variety of $c$ values and root percentages from 15\% to 30\%,
generally finding the optimal solution. In many of these cases KEP also finds
an optimal solution but does not terminate and confirm that the solution is
optimal so PCD has a significant advantage in running time in those cases. BP
and Greedy are suboptimal in these high-$K$ cases, both finding 80\%-90\% of
the nodes in the optimal solution.

\subsection{Computational results for real world networks.}

We compared the performance of KEP, PCD, Greedy, and BP algorithms on several
real world networks taken from the Stanford Large Network Dataset collection
\cite{snapnets}. We chose directed graphs and randomly selected a subset of
nodes to act as roots (discarding directed edges pointing in to roots and
discarding isolated nodes). We present performance comparisons for five
networks, which are as follows:
\begin{itemize}
\item Epinions: this network is drawn from the website Epinions.com and
represents ``trust'' relationships between users, in which each user can
indicate whether they trust the opinions of each other user. Thus an edge from
$i$ to $j$ indicates that user $i$ trusts user $j$.	
\item Gnutella: a snapshot of the Gnutella peer-to-peer file sharing network,
with nodes representing hosts and edges representing the connections between
those hosts.
\item Slashdot: users of the technology news site Slashdot were allowed to tag
other users as friends or foes; our network has users as nodes and tags as
directed edges, drawn from a snapshot in 2008.
\item Wiki-vote: A record of votes on administration positions for Wikipedia,
in which administrators vote on the promotion of other potential
administrators. Nodes are wikipedia users and a directed edge from node $i$ to
node $j$ represents that user $i$ voted on user $j$.
\item Amazon: product co-purchasing on Amazon.com, based on the ``customers who
bought this item also bought feature''. A directed edge from product $i$ to
product $j$ indicates that $j$ appears on the page for $i$ as a frequently
co-purchased item. The data is a snapshot from March 2003.
\end{itemize}

\begin{table}
	\centering
	\begin{tabular}{|c c c| c c c |}
		\multicolumn{3}{c}{} & \multicolumn{3}{c}{Nodes in solution} \\
		\cline{1-6}
Graph & Nodes & $\Delta$ & KEP & Greedy & BP \\
		\cline{1-6}
Epinions & 68507 & 6.0 & - & 21,219.4 & \cellcolor{best} 24,347.4 \\
Gnutella & 5518 & 3.0 & \cellcolor{best} 1,889.6 & 1,610.8 &
\cellcolor{better} 1,708.2 \\
Slashdot & 74893 & 9.7 & - & \cellcolor{best} 36,770.8 & 31,266.8 \\
Wiki-vote & 6439 & 12.9 & - & 2,045.0 & \cellcolor{best} 2,173.0 \\
Amazon & 260982 & 3.8 &  - & 174,035.8 & \cellcolor{best} 180,902.6 \\
		\cline{1-6}
		\multicolumn{6}{c}{} \\
		\multicolumn{3}{c}{} & \multicolumn{3}{c}{Running Time (min)} \\
		\cline{1-6}
Graph & Nodes & $\Delta$ & KEP  & Greedy & BP \\
		\cline{1-6}
Epinions & 68507 & 6.0 & 56.5 & 9.8 & 17.9  \\
Gnutella & 5518 & 3.0 &  150.3 & 0.2 & 0.4  \\
Slashdot & 74893 & 9.7 &  58.9 & 5.1 & 20.5  \\
Wiki-vote & 6439 & 12.9 & 75.5 & 0.8 & 1.5  \\
Amazon & 260982 & 3.8 & 1,042.8 & 182.0 & 259.9  \\
		\cline{1-6}
	\end{tabular}
	\caption{Real world networks with 20\% root nodes (randomly selected) and
maximum path length $K = 5$. We use $\beta = 0.01$ for BP. $\Delta$ is the
average degree of the sampled graph. 5 samples were run for each graph, with
$-$ representing graphs for which KEP crashed without returning any solution.
\colorbox{best}{Shaded cells} represent the best solutions, with
\colorbox{better}{lighter shading} for the case when BP provides a better
solution than Greedy.}
	\label{fig:real} 
\end{table}

The results of our tests on the full graphs are summarized in Table
\ref{fig:real}. The ``Nodes'' column indicates the size of the given graph,
which is given as an average over the five tests because the selection of root
nodes results in a different number of nodes becoming isolated and being
dropped in each test. The column ``$\Delta$'' is the average degree of the
graph, which is also averaged across the choices of root nodes for each trial.
Remaining columns show the number of nodes in the average solution and the
running time in minutes for each algorithm. With the exception of the
graph Gnutella, KEP in all cases crashed rather than running to the stopping
time and so provided no solution, which we indicate by ``-''.  We believe these
crashes resulted from KEP exceeding the available memory. They also caused the
running times for KEP to be erratic, with the algorithm crashing before the
time limit in most cases. The extreme running time for Amazon, the largest
graph, can be explained by the external software we used failing to adhere to
the requested time limits before crashing. In all cases PCD failed to return a
nontrivial solution, so it is omitted from the table. For all graphs except
Slashdot, BP enjoys an advantage in solution size over Greedy, finding 4\%-15\%
more nodes while running only slightly longer on average. On Slashdot Greedy
has an advantage in both solution size and running time. We also note that even
on the graph Gnutella, the one case where KEP returned feasible solutions, BP
has a huge advantage in running time for a solution with 90\% of the nodes of
that returned by KEP.

\bibliographystyle{plain}
\bibliography{mp}

\end{document}